\documentclass[11pt]{article}
\usepackage[preprint]{neurips_2020}
\usepackage[utf8]{inputenc}

\usepackage{bbm}
\usepackage[title]{appendix}
\usepackage{subcaption}
\newcommand{\supp}{\operatorname{supp}}
\newcommand{\bpi}{\boldsymbol{\pi}}
\usepackage{dsfont}
\usepackage{amsfonts}
\usepackage{amsmath}
\usepackage{amssymb}
\usepackage{amsthm}
\DeclareMathOperator*{\argmax}{arg\,max}

\usepackage{booktabs}
\usepackage{multicol}
\usepackage{multirow}
\usepackage{tabularx}
\usepackage{diagbox}
\usepackage{graphicx}
\usepackage{array}
\usepackage{fancyhdr}
\usepackage{hyperref}
 \hypersetup{
     colorlinks=true,
     linkcolor=blue,
     filecolor=blue,
     citecolor = black,      
     urlcolor=cyan,
     }
\usepackage{cleveref}
\usepackage{xcolor}
\usepackage{tikz}
\usepackage{verbatim}
\usepackage{listings}
\definecolor{mygreen}{RGB}{28,172,0} % color values Red, Green, Blue
\definecolor{mylilas}{RGB}{170,55,241}

\usepackage{matlab-prettifier}
\usepackage{color}
\usepackage[shortlabels]{enumitem}

\usepackage{algorithm}
\usepackage[noend]{algpseudocode}
\makeatletter
\def\BState{\State\hskip-\ALG@thistlm}
\makeatother
\newcommand{\pr}{\mathbb{P}}
\newcommand{\E}{\mathbb{E}}
\newcommand{\tp}{\mathsf{T}}

\newcommand{\N}{\mathbb{N}}
\newcommand{\R}{\mathbb{R}}

\newcommand{\st}{\text{subject to}}
\newtheorem{theorem}{Theorem}
\newtheorem{lemma}{Lemma}
\newtheorem{corollary}{Corollary}
\newtheorem{proposition}{Proposition}

\newtheorem{definition}{Definition}

\title{Commitment with Signaling under Double-sided Information Asymmetry}
\author{Tao Li\thanks{Corresponding author}\quad  Quanyan Zhu\\Department of Electrical and Computer Engineering\\ New York University  \\ \texttt{\{taoli, qz494\}@nyu.edu} 
}
\begin{document}
\maketitle
\begin{abstract}
Information asymmetry in games enables players with the information advantage to manipulate others' beliefs by strategically revealing information to other players. This work considers a double-sided information asymmetry in a Bayesian Stackelberg game, where the leader's realized action, sampled from the mixed strategy commitment, is hidden from the follower. In contrast, the follower holds private information about his payoff. Given asymmetric information on both sides, an important question arises: \emph{Does the leader's information advantage outweigh the follower's?} We answer this question affirmatively in this work, where we demonstrate that by adequately designing a signaling device that reveals partial information regarding the leader's realized action to the follower, the leader can achieve a higher expected utility than that without signaling. Moreover, unlike previous works on the Bayesian Stackelberg game where mathematical programming tools are utilized, we interpret the leader's commitment as a probability measure over the belief space. Such a probabilistic language greatly simplifies the analysis and allows an indirect signaling scheme, leading to a geometric characterization of the equilibrium under the proposed game model.	
\end{abstract}

\section{Introduction}
\label{sec:intro}
Information asymmetry generally accounts for where one party acquires more information than the rest. By deliberately revealing information unknown to others, the party with the information advantage can manipulate the belief of receivers to achieve a higher utility. Information asymmetry has long existed since the market came into being. It has been of great significance in game-theoretic studies, such as mechanism/information design and contract theory \citep{fudenberg,bolton2005contract, yang2023designing}, game-theoretic learning \citep{tao_confluence,tao_info}, network security and control \citep{pan-tao22noneq, liu-li23,ge2023scenario} and adversarial machine learning \citep{pan2023first}. In this paper, we study a doubled-sided information asymmetry between the leader and the follower in the context of Stackelberg games, as the game model has gained much popularity in the realm of security research and has been widely applied in real-world applications, including wildlife protection \citep{yang14wild} and security systems \citep{tambe11ssg}. 

In a basic Bayesian Stackelberg game (BSG) model \Cref{sec:model}, the information asymmetry is one-sided, where the leader has incomplete knowledge about the follower's payoff. Hence, based on the prior distribution of the follower's type, the leader would commit to the optimal mixed strategy that brings her the highest expected utility\footnote{Following the convention, we refer to the leader as a female and to the follower as a male.}. As is often assumed in BSG \citep{tambe11ssg}, the leader's realized action (pure strategy) cannot be observed by the follower, implying the follower can only best respond according to the mixed strategy announced by the leader. However, this hidden action also brings the leader asymmetric information over the follower, and a natural question arises: Can the leader leverage this information to get better off, compared with her Bayesian Stackelberg equilibrium payoff? 

In this paper, we answer this question affirmatively by showing that with the first-mover advantage, the leader can utilize her asymmetric information when dealing with the unknown follower. Specifically, we proposed a new class of Stackelberg games based on the baseline BSG. We refer to the proposed model as a signaling Bayesian Stackelberg game (\texttt{sigBSG}). In this model, the leader can include a signaling device as part of her commitment, signaling differently to followers of different types. It is straightforward that the leader can inevitably benefit from her extra commitment ability, i.e., designing a signaling device. However, it should be noted that the follower's asymmetric information makes the gameplay more involved, as he can also leverage his information advantage to influence the leader. To better illustrate the double-sided information asymmetry and strategic interactions under this circumstance, we use a running example in \Cref{sec:example}.

Our contributions are two-fold. 1) We propose a new class of Stackelberg game (\texttt{sigBSG}) models under doubled-sided information asymmetry. By showing that \texttt{sigBSG} acquires the same information structure, we show that Bayesian Stackelberg equilibrium (BSE) is still a proper solution concept in \texttt{sigBSG}, though the equilibrium is structurally much more complicated; 2) To characterize the sigBSE, thanks to the randomness created by the signaling device, we can interpret the leader's commitment as a probability measure over the belief space, which leads to a geometric characterization of sigBSE and a finite-time algorithm for equilibrium seeking. As a direct corollary of the geometric characterization, one interesting finding is that even though the leader does offset the follower's information advantage, achieving a higher payoff, the follower may render the leader's exact equilibrium strategy unattainable under some circumstances, implying that the follower's asymmetric information also influences the leader in the gameplay.  

\section{A Running Example}
\label{sec:example}
The Stackelberg game model was initially proposed by von Stackelberg \citep{stackelberg} when studying market competition between a leader, such as a leading firm in the market, and a follower, which for example, is an emerging start-up. Due to its dominance, the leader enjoys the first-mover advantage, moving first or announcing a commitment before the follower makes decisions. In order to capture the leader's incomplete knowledge about the follower's payoff, the  Bayesian model is considered, where the type of the follower is drawn from a finite set according to a prior distribution. Then, the follower's information asymmetry is that his type is private, which is hidden from the leader. In the following, we provide a concrete example adapted from \citep{hifeng16signaling,gan19imitative} to illustrate how the leader, on the other hand, can leverage her asymmetric information about her covert action. 

Consider a market competition between a leading firm (the leader) and a start-up (the follower). The leader specializes in two products: product 1 and product 2. Due to its limited capacity, the follower can only focus on one product, and it is known that the follower will choose product 1 with a probability of $0.55$ and product 2 with a probability of $0.45$. To fit the Bayesian model, we refer to the follower choosing product 1 as the follower of type $\theta_1$ and the other one as the follower of type $\theta_2$. As mentioned above, the type is only revealed to the follower himself. In this competition, the leader has to decide which product she devotes to drive the follower out of the market. The follower has two options: enter the market and sell the product he focuses on or leave the market. \Cref{tab:payoff} summarizes the payoff structures, where the leader's action denotes the product she devotes to; $i_1$ for product 1, $i_2$ for product 2, and $i_0$ means she maintains the current production. Similarly, the follower's action $j\in \{j_0,j_1\}$ denotes whether he enters the market $j_1$ or leave $j_0$.    
\begin{table}
\caption{Payoff Structures of the Bayesian Stackelberg game. The leader only gets 1 when the follower leaves the market, whereas the follower only gets a positive utility when he successfully survives: entering the market and avoiding competing with the leader. The following two tables give the payoff matrices of different followers: the row player is the leader, and the column player is the follower.  }
\begin{subtable}{.5\linewidth}
      \centering
        \caption{$\theta_1$ follower }
      \centering
        \begin{tabular}{|c|c|c|}
\hline
\diagbox{L}{F}& $j_0$ & $j_1$  \\ \hline
$i_0$ & (1,0) & (0,2)  \\ \hline
$i_1$ & (1,0) & (0,-1) \\ \hline
$i_2$ & (1,0) & (0,2)  \\ \hline
\end{tabular}
    \end{subtable}%
    \begin{subtable}{.5\linewidth}
      \centering
         \caption{$\theta_2$ follower}
    \begin{tabular}{|c|c|c|}
\hline
\diagbox{L}{F} & $j_0$ & $j_1$  \\ \hline
$i_0$ & (1,0) & (0,1)  \\ \hline
$i_1$ & (1,0) & (0,1) \\ \hline
$i_2$ & (1,0) & (0,-1)  \\ \hline
\end{tabular}
    \end{subtable}
    \label{tab:payoff}
\end{table}

As shown in \Cref{tab:payoff}, when the follower leaves the market, the leader gets 1; otherwise, 0, and hence the leader's objective is to drive the follower out of the market. On the other hand, the follower's goal is to survive in the market by avoiding competition with the leader. From the payoff matrices, to defeat the follower of type $\theta_1$, the leader has to commit to a mixed strategy that allows her to play $i_1$ at least with probability $2/3$. Similarly, to force the type $\theta_2$ to leave the market, the leader has to play $i_2$ with a probability of at least $1/2$. There does not exist a mixed strategy that drives both types out of the market, and the optimal leader's strategy in this Bayesian Stackelberg game is to play $i_1$ at least with probability $2/3$, which makes the type $\theta_1$ leaves the market, as the follower of this type occurs with a higher probability. The leader expected utility is $0.55$ under this Bayesian Stackelberg equilibrium (BSE). Without loss of generality, we assume that the follower breaks the ties in favor of the leader. The assumption is common in Stackelberg games \citep{fudenberg}. Generally, the leader can induce the follower to best respond to the leader's benefit by infinitesimal strategy variations. The corresponding solution concept is referred to as the strong Stackelberg equilibrium \citep{tambe11ssg}.

We now consider the signaling commitment, where the leader signals differently to different followers. Even though the leader is unaware of the follower's private type,  the leader can simply ask the follower to report his type. We first demonstrate that if the follower truthfully reports, the leader's expected utility increases. When using the signaling commitment, the leader announces her mixed strategy $(0,1/2,1/2)$ and also a signaling scheme specified in the following. There are two different signals $\{s_0,s_1\}$ at the leader's disposal. When facing a type $\theta_1$ follower, the leader sends $s_0$ to the follower, when the realized action is $i_0$ and will send $s_1$ uniformly at random when $i_1$ is implemented. On the other hand, for the follower of type $\theta_2$, the leader will always send $s_0$ whatever the realized action is. Formally, this signaling scheme can be described by the following conditional probabilities:
\begin{align*}
\begin{array}{ll}
    \pr(s_0|i_1,\theta_1)=1 &\pr(s_1|i_1,\theta_1)=0;\\
    \pr(s_0|i_2,\theta_1)=\frac{1}{2}& \pr(s_1|i_2,\theta_1)=\frac{1}{2};\\
    \pr(s_0|i_1,\theta_2)=1 &\pr(s_1|i_1,\theta_2)=0;\\
    \pr(s_0|i_2,\theta_2)=1 &\pr(s_1|i_2,\theta_2)=0.
\end{array}
\end{align*}

First, we notice that this signaling is totally uninformative to the type $\theta_2$, as he only receives the signal $s_0$, and the posterior belief produced by the Bayes rule is exactly the same as the announced mixed strategy. As a result, his best response is to leave the market. On the other hand, the follower of type $\theta_1$ receives $s_0$ with probability 3/4, based on which, he infers that the leader's realized action is $i_1$ with probability $\pr(i_1|s_0,\theta_1)=2/3$ or $i_2$ with probability $\pr(i_2|s_0,\theta_1)=1/3$. Holding this belief, the follower would choose to leave the market, assuming he breaks ties in favor of the leader. A similar argument also applies to the case where the follower of $\theta_1$ receives $s_1$, which makes the follower believe  that leader is taking $i_2$ with probability $\pr(i_2|s_1,\theta_1)=1$. Hence, he will enter the market. As a result, with this signaling commitment, the leader is able to drive the follower of $\theta_2$ out of the market for sure and type $\theta_1$ follower out of the market 3/4 of the time, yielding an expected utility $0.8625>0.55$ higher than that achieved by a BSE strategy.  
 
 However, the success of the signaling commitment relies on the leader distinguishing different follower types and signaling differently. With incomplete knowledge, the leader needs the follower to report his type, and this is where the follower can leverage his information asymmetry. As shown in the discussion above, once the leader announces the signaling commitment, the follower of $\theta_2$ would immediately realize that the signaling does not reveal any information had he truthfully reported the type. Instead, it is in his interest to misreport this type as $\theta_1$, which allows him to infer the leader's realized action: when he receives $s_0$, he will enter the market, as the leader is more likely to be focusing on product 1, whereas he leaves the market when receiving $s_1$, as he knows that the leader devotes to product 2 for sure. By misreporting, the type-$\theta_2$ follower may find a chance to survive, increasing his expected utility. Finally, we note that the follower of $\theta_1$ is not incentivized to deceive the leader, since his original signals are more informative than what he would obtain by misreporting. 

\subsection{Related Works}
  This motivating example demonstrates that how additional commitment power can benefit the leader by signaling to followers, resulting in a higher expected payoff than a BSE. With the picture of the potential benefit of signaling commitment as well as much involved strategic interactions due to double-sided information asymmetry, we briefly review some related existing works, showing the connections and the novelty of our work.  
\paragraph{Signaling Commitment} Our work is inspired by the correlation commitment studied in    \citep{korzhyk11correlated,hifeng16signaling}. The signaling device in our model can also be viewed as a correlation device, which coordinates players' actions in the desired way (desired by the leader). Since the correlation device is designed by the leader, it never hurts the leader (by sending uninformative signals), instead, it can increase the leader's payoff, as shown in the example. However, similar to other early works on Stackelberg games \citep{conitzer06compute_SG,letchford09br_learn,marecki12MCTS_SG}, the research works mentioned above focus on the algorithmic aspect, where the computation of the equilibrium is of main interest. Therefore, in the works mentioned above, the follower is assumed to be truthfully behaving all the time. In contrast, in our model, the follower is a rational player, whose strategic behaviors are taken into account. 

\paragraph{Follower Imitative Deception} Our work also subscribes to the recent line of works that explore the information asymmetry in Stackelberg games due to the leader's incomplete knowledge, and the resulting follower's behaviors. As discussed in \citep{gan19imitative,gan19counter,gan20opt_deceive}, a strategic follower is incentivized to leverage the information asymmetry by behaving as if he is of another type, which the authors refer to as the imitative deception. As shown in \citep{gan19counter}, the follower would behave in such a way that makes the leader believe the game is zero-sum, misleading the leader to play her maxmin strategy, which means that the leader does not obtain any useful information regarding the follower's payoff through the interaction. Our work differs from these works in that the double-sided information asymmetry is considered in our model, where the leader also enjoys an information advantage. Furthermore, we show that even under this follower's imitative deception, with her asymmetric information and the first-mover advantage, the leader still benefits from the signaling commitment.    
\paragraph{Information Design} Finally, our work is also related to a much broader context in game theory: information design \citep{emir11persuasion,kolo17private,matteo20online_pers}, as the leader considered in this model not only commits to a mixed strategy but also designs a signaling device for revealing information. As indirect information design is considered in this work, our problem becomes more challenging in that the leader must take into account the possible misreports, which makes the leader's optimization problem (see \eqref{eq:direct_opt} ) complicated.  
\section{Signaling in Bayesian Stackelberg Games}\label{sec:sigBSG}
In this section, we first give a generic model of the Bayesian Stackelberg game studied in this paper, and then we discuss the proper solution concept, considering the information structure of the game. The mathematical description of the equilibrium of \texttt{sigBSG} is given in the \Cref{sec:probab} after we introduce the probabilistic language for describing the model. Throughout the paper, we use the following notations. Given an integer $n\in\N$, let $[n]:=\{1,2,\ldots,n\}$ be the set of all natural numbers up to $n$.  Given a set $\mathcal{X}$ in a Euclidean space, let $\Delta(\mathcal{X})$ denote the set of Borel probability measures over $\mathcal{X}$. Specifically, we denote $\Delta([n]):=\{x\in \R^n| \mathbf{1}^\tp x=1\}$ by $\Delta_n$, which is a probability simplex in $\R^n$. Finally, we denote by $\mathcal{X}^n$ the product space $\prod_{k=1}^n \mathcal{X}$. 

\subsection{The Model}\label{sec:model}
We start with a basic Bayesian Stackelberg game with one leader type and multiple follower types. Let $\Theta:=\{\theta_1,\theta_2,\ldots, \theta_K\}$ be the set of all the follower types, and its cardinality is $|\Theta|=K\in \N$. The type $\theta$ follower occurs with a probability $\mu^*(\theta)$, where $\mu^*$ is a probability measure over $\Theta$. We use $[M], M\in \N$ and $[N], N\in \N$ to denote the leader's and follower's pure strategy set, respectively. The payoff of the leader is determined by the utility function $L:[M]\times[N]\rightarrow \R$, and similarly, the utility function of type $\theta$ follower is $F^\theta:[M]\times[N]\rightarrow \R$. We assume that both players' utility functions are bounded. In the basic model of Bayesian Stackelberg game, the leader would first commit to a mixed strategy $\mathbf{x}\in \Delta_M$. It should be noted that the mixed strategy is not mere randomization over pure strategies. Instead, it leads to a stronger commitment advantage \citep{conitzer16mixed_sg} and has attracted significant attention due to its direct applications in real security domains \citep{tambe11ssg}.  After observing the leader's mixed strategy, the follower would come up with the best response, maximizing his expected utility under the mixed strategy. Without loss of generality, we assume the follower's best response is a pure strategy with ties broken in favor of the leader.  Finally, the leader plays a pure strategy sampled from the mixed strategy, while the follower plays his best response, and their payoffs are determined by their utility functions. We note that the leader's realized action is effectively hidden from the follower since the follower's best response only depends on the mixed strategy.

Apart from the basic setup, we consider an even stronger commitment: signaling commitment, which consists of a mixed strategy as well as a signaling scheme, which sends a signal to the follower based on the leader's realized action as well as the type reported by the follower. Since the leader acquires additional commitment ability, the gameplay should be revised accordingly. Compared with the three stages in the basic Stackelberg game, in the proposed model, the gameplay comprises four stages: 1) the leader announces her signaling commitment, including her mixed strategy and the signaling scheme; 2) the follower reports his type to the leader, in order to receive the signal; 3) based on her realized action, as well as the reported type, the leader, sends a signal the follower; 4) after receiving the signal, the follower updates his belief about the leader's realized action, based on which, he best responds. The payoffs of both sides are determined by the players' pure strategies according to their utilities.  To better present our analysis, in the following, we use $\hat{\theta}$ to denote the reported type, and $(\hat{\theta};\theta)$ indicates that the follower of true type $\theta$ reports $\hat{\theta}$ to the leader.

We now describe the signaling scheme in mathematical terms. Let $\mathcal{S}$ be a finite set of signals, from which the leader draws a signal $s\in \mathcal{S}$ according to a conditional distribution specified by the scheme. Mathematically, a signaling scheme is a mapping $\phi:[M]\times\Theta\rightarrow \Delta(\mathcal{S})$, which determines the probability of sending signal $s$ conditional on the realized action $i\in [M]$ and the reported type $\hat{\theta}\in \Theta$, that is $\pr(s|i,\hat{\theta})=\phi(s|i,\hat{\theta})$. To simplify our exposition, by revelation principle in Bayesian persuasion \citep[Proposition 1]{emir11persuasion}, we consider the direct signaling device, for which the signal space coincides with the follower's action space, that is $\mathcal{S}=[N]$. Under this signaling, each signal can be interpreted as an action recommendation. In the example in \Cref{sec:example}, for the type $\theta_1$ follower, the signal $s_0$ can be replaced by $j_0$, suggesting that the follower would better leave the market since the leader is devoting resources to the corresponding product. Similarly, $s_1$ is equivalent to $j_1$, which means the leader encourages the follower to enter the market. However, we note that this is merely an interpretation, and each signal, except inducing a posterior belief, does not carry a meaning beyond itself. 

Given a signaling commitment $\sigma:=(\mathbf{x},\phi)$, when the type $\hat{\theta}$ is reported, the probability of the leader playing  $i$ and recommending the follower to play $j$ is $\pr(i,j|\theta)=\mathbf{x}_i\phi(j|i,\theta)$, meaning that under the reported type $\hat{\theta}$, the commitment $\sigma$ in fact is equivalent a correlation strategy \citep{korzhyk11correlated}, which we denote by $C^\theta_{ij}:=\pr(i,j|\theta)$. Hence, the signaling commitment can be redefined as $\sigma=(\mathbf{x},\mathbf{C})$, where $\mathbf{C}\in [0,1]^{M\times N\times K}$ and its $(i,j,\theta)$ component $C^\theta_{ij}$ specifies the correlation. We will use this new definition of the signaling commitment, which is essentially the same as the original definition, since from our argument above, the original signaling scheme $\phi$ can be easily recovered from $\mathbf{x}$ and $\mathbf{C}$. 

\subsection{The Solution Concept}
One challenge in analyzing the proposed game model is the double-side information asymmetry. The leader's realized action is unknown to the follower, and hence, by signaling, the leader can manipulate the follower's belief. On the other hand, the leader's incomplete knowledge also brings information advantage to the follower, who may take advantage of the signaling by misreporting, as shown in the running example. This double-sided information asymmetry raises a fundamental question: \emph{Can the leader offset the follower's information advantage by carefully designing the signaling commitment, which can bring her higher expected utility than that of BSE?} To answer this question, we first need to characterize the leader's optimal signaling commitment in the proposed \texttt{sigBSG}. 

Before proceeding to the detailed discussions, we first revisit the \texttt{sigBSG} model. To properly define the solution concept of this three-stage dynamic game (with the chance move included), we present the information structure of \texttt{sigBSG} in the following (see \Cref{fig:game_tree}). The key is that by viewing the leader's signaling commitment as her strategy and the follower's optimal reporting as the best response, \texttt{sigBSG} can be viewed as an extension of BSG, and hence, the notion of Bayesian Stackelberg equilibrium (BSE) also applies to \texttt{sigBSG}. To be specific, when the leader decides her signaling commitment, what she knows only includes the common knowledge of the game: the type space $\Theta$ and associated distribution $\mu$, as well as the utility functions of both sides. After the announcement, the follower needs to decide his reporting strategy based on the signaling and his true type. We note that this reporting strategy is purely based on the follower's expected utility with respect to the announced signaling, the reported type as well the true type, since by the time the follower reports, he knows nothing about the leader's move, except the signaling commitment. Therefore, \texttt{sigBSG} acquires the same information structure as BSG does, as shown in \Cref{fig:game_tree}. To better illustrate this similarity, we denote by $V(\sigma,(\hat{\theta};\theta))$ the leader's expected utility under the signaling commitment $\sigma$ and the reporting $(\hat{\theta};\theta)$. While for the follower of type $\theta$, we denote his expected utility by $U(\sigma,(\hat{\theta};\theta))$. All the involved quantities and mappings will be rigorously defined in \Cref{sec:probab}, and here we aim to provide the reader with a bird's-eye view of the leader's problem, which can be written as 
\begin{align}\label{eq:direct_opt}
\begin{aligned}
	\max_{\sigma=(\mathbf{x},\mathbf{C})} & \quad\sum_{\theta\in \Theta}\mu(\theta) V(\sigma,(\hat{\theta};\theta)),\\
	\st&\quad \hat{\theta}=\argmax_{\beta\in \Theta}U(\sigma, (\beta;\theta)).
\end{aligned}
\end{align}              
Therefore, from our argument above, it can be seen that in our proposed \texttt{sigBSG}, a proper solution concept is still the Bayesian Stackelberg equilibrium (BSE), as the follower can perfectly observe the leader's commitment. Compared with vanilla BSE, the follower determines his optimal reporting according to the expected utility $U(\sigma,(\hat{\theta};\theta))$, which interestingly raises the issue of the nonexistence of exact equilibrium, as we shall see it more clearly at the end of the following section. 
\begin{figure}
	\centering
	\includegraphics[width=0.6\textwidth]{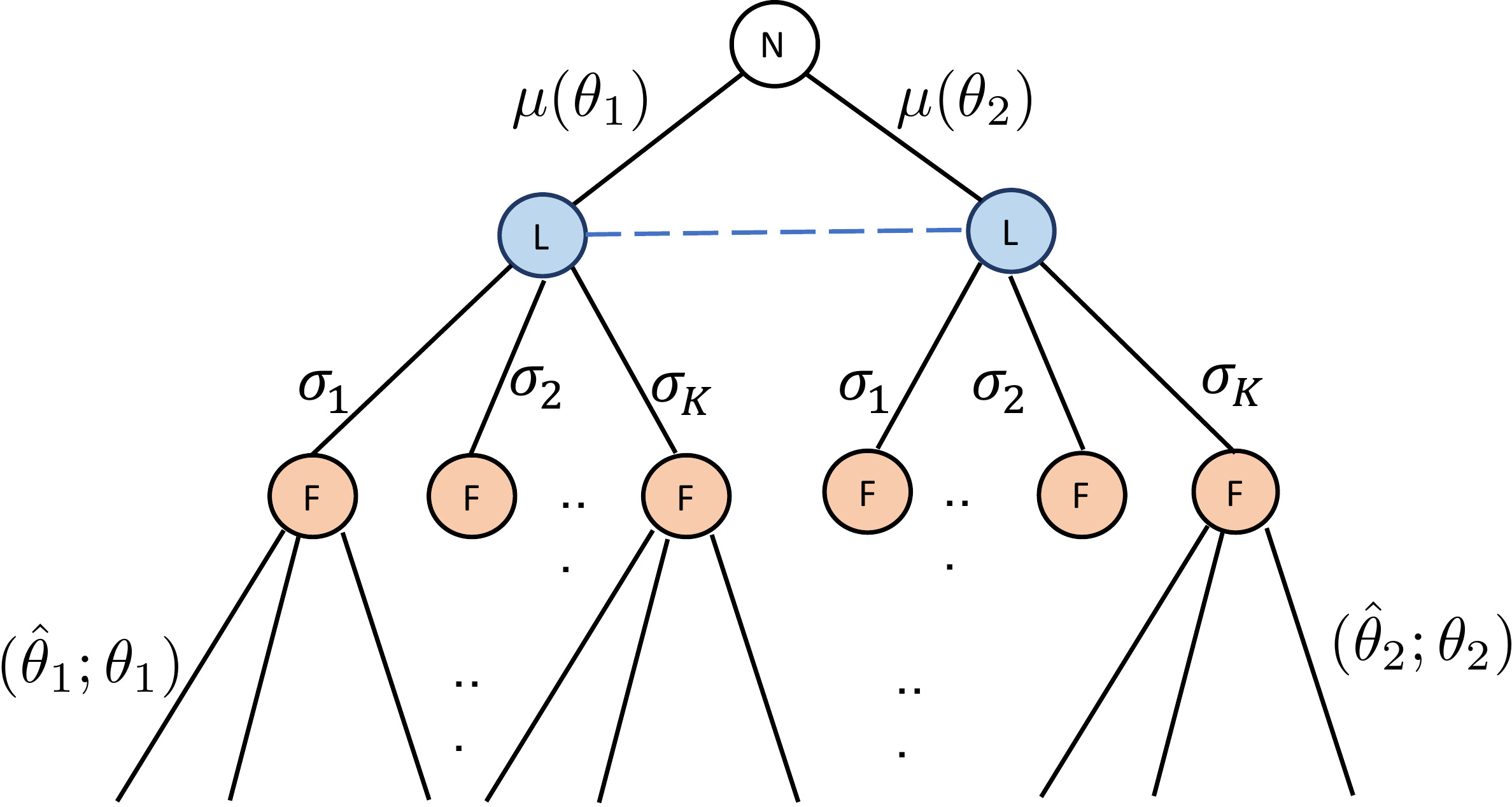}
	\caption{The game tree of \texttt{sigBSG}. Nature first reveals the follower's type according to distribution $\mu$. Then, unaware of the private type (denoted by the blue dashed line), the leader announces the signaling commitment $\sigma$. After observing the commitment, the follower decides what type to report. The rest of the play is suppressed because with the signaling $\sigma$ and reporting $(\hat{\theta};\theta)$, the expected utility is already determined.}
	\label{fig:game_tree}
\end{figure}

\section{Equilibrium Characterization}\label{sec:probab}
Given the leader's commitment $\sigma=(\mathbf{x},\mathbf{C})$,  the follower considers what he could obtain if he misreports. Denote $\nu^{\hat{\theta}}_j:=\sum_{i\in [M]}C^{\hat{\theta}}_{ij}$ the probability of the leader sending signal $j$, when the follower reports $\hat{\theta}$. When the leader recommends $j$, the follower can infer that the leader is playing $i$ with probability $C^{\hat{\theta}}_{ij}/\nu^{\hat{\theta}}_j$. Then, the follower's best response when receiving $j$ is
\begin{align}\label{eq:br_hard}
	j':=\argmax_{\tilde{j}}1/\nu^{\hat{\theta}}_j\sum_{i\in [M]}C^{\hat{\theta}}_{ij}F_{i,\tilde{j}}.
\end{align}
 Therefore, the follower's expected utility when he misreports is given by
$$U(\sigma,(\hat{\theta};\theta)):=\sum_{j\in [N]}\bigg[\max_{j'}\sum_{i\in [M]}C^{\hat{\theta}}_{ij}F_{ij'}\bigg],$$ which further leads to the follower's optimal reporting strategy 
\begin{align}\label{eq:report}
	\psi_{\theta}(\sigma):=\argmax_{\hat{\theta}\in \Theta}U(\sigma,(\hat{\theta};\theta)).
\end{align}
Finally, the leader's expected utility under the misreport is given by 
\begin{align*}
	 V(\sigma,(\psi_{\theta}(\sigma);\theta))=\sum_{i\in [M], j\in [N]}C^{\psi_{\theta}(\sigma)}_{ij} L_{i,j'},
\end{align*}
where $j'$ is given by \eqref{eq:br_hard}. We note that with the above discussion, it can be easily seen that directly solving \eqref{eq:direct_opt} is extremely difficult. Due to the interdependence of the leader's commitment and the follower's best response, the objective function is not linear nor convex and not even a continuous function with respect to $\mathbf{x}, \mathbf{C}$. Therefore, directly solving \eqref{eq:direct_opt} as in previous works \citep{conitzer06compute_SG,hifeng16signaling,korzhyk11correlated} is not viable in the proposed \texttt{sigBSG}, and we need another approach to tackling this problem. Inspired by the probabilistic viewpoint adopted in Bayesian persuasion \citep{emir11persuasion,matteo20online_pers}, in the following, we rewrite the signaling commitment as a probability measure, which greatly simplifies the problem. 

From the leader's perspective, given a signaling commitment $\sigma$, when the reported type is $\hat{\theta}$, and  the signal $j\in [N]$ is sent,  the following bayesian update should be performed by the follower to infer a posterior belief $b_j(i):={C}^{\hat{\theta}}_{ij}/\sum_{i\in [m]}{C}^{\hat{\theta}}_{ij}$ about the leader's realized action.
Then under this signaling commitment, the probability of incurring $b_j$ is given by 
\begin{align*}
	\pi^{\hat{\theta}}(b):=\sum_{j\in [N]: b_j=b}\sum_{i\in[M]}{C}^{\hat{\theta}}_{ij}.
\end{align*}  
In other words, $\pi^{\hat{\theta}}$ is the distribution over posterior beliefs when the follower reports $\hat{\theta}$ under the signaling. Moreover, from the definitions of $b$ and $\pi^\theta$, we have $\sum_{b\in\supp(\pi^\theta)}\pi^\theta(b)b(i)=\mathbf{x}_i$, which implies that the expected posterior beliefs equal the mixed strategy. Let $\Xi$ be the space of all beliefs, that is $\Xi:=\Delta_M$, and we denote by $\boldsymbol{\pi}:=[\pi^{\theta_1},\pi^{\theta_2},\ldots,\pi^{\theta_{K}}]$ the product probability measure over the product of $K$ belief spaces, that is $\boldsymbol{\pi}\in (\Delta(\Xi))^K$. $\Pi$ is said to be a set of consistent probability measure if for every $\theta\neq \theta'\in \Theta$, $$\sum_{b\in\supp(\pi^\theta)}\pi^\theta(b)b(i)=\sum_{b\in\supp(\pi^{\theta'})}\pi^{\theta'}(b)b(i).$$

This consistency implies that a valid signaling commitment should ensure that every follower's expected posterior belief is the same, as the leader can only commit to one mixed strategy. We denote by $\Pi\subset \left(\Delta(\Xi)\right)^K$  the set of probability measures that are consistent. Finally, we point out that the $\bpi$ and $\sigma$ are equivalent, and the two will be used interchangeably\footnote{The equivalence has been shown in \citep{emir11persuasion}, and the notion similar to the consistency defined in this paper is referred to as bayesian plausibility in the bayesian persuasion literature \citep{emir11persuasion}}. 

From the follower's perspective, after observing a signal that induces a posterior belief $b\in \Xi$, the follower best responds by choosing the action that maximizes his expected utility under this belief. Formally, given a belief $b$ and the type $\theta$, the best response set is 
\begin{align*}
	BR^\theta(b):=\argmax_{j\in [N]}\sum_{i\in [M]}b(i)F^\theta(i,j).
\end{align*} 
Without loss of generality, we assume that the best response set is always a singleton, as the follower can be induced to break the tie in favor of the leader. Under the belief $b$, with a slight abuse of notations, the expected utility of the follower is defined as
	$F^\theta(b):=\sum_{i\in [M]}b(i)F^\theta(i,BR^\theta(b))$.
Therefore, under the signaling commitment, the type $\theta$ follower's expected utility, if he reports $\hat{\theta}$, is defined as
\begin{align*}
	U(\bpi, (\hat{\theta};\theta)):=\sum_{b\in \supp(\pi^{\hat{\theta}})}\pi^{\hat{\theta}}(b)F^\theta(b).
\end{align*}
For any consistent $\bpi$, if there exists $\hat{\theta}$ such that $U(\bpi,(\hat{\theta};\theta))>U(\bpi,(\theta;\theta))$, for some $\theta\in \Theta$, then under the signaling commitment, the follower of type $\theta$ is incentivized to misreport his type.  Following this line of reasoning, we can rewrite the definition of the follower's optimal reporting strategy in \eqref{eq:report} as
\begin{align}\label{eq:new_report}
    \psi_{\theta}(\bpi):=\argmax_{\hat{\theta}\in \Theta}U^\theta(\bpi,(\hat{\theta};\theta)).
\end{align}
As for tie-breaking, we assume the follower can choose some arbitrary but consistent rule. For example, the reported type should be the $\theta_k$ with the smallest index. 

Similarly, for the leader's expected utility under the signaling commitment, we can introduce the following notations. Given a belief $b$, and the follower's true type $\theta$, we denote by $V^\theta(b):=\sum_{i\in [M]}b(i)L(i,BR^\theta(b))$ the leader's expected utility when she induces the belief $b$. Moreover, we denote by $V(\boldsymbol\pi, (\hat{\theta};\theta))$ the leader's expected utility achieved under the signaling scheme $\sigma$(equivalently $\boldsymbol{\pi}$) when the the follower of type $\theta$ reports $\hat{\theta}$, which is defined as 
\begin{align*}
	V(\boldsymbol{\pi},(\hat{\theta};\theta))=\sum_{b\in \supp(\pi^{\hat{\theta}})}\pi^{\hat{\theta}}(b)V^{\theta}(b),
\end{align*}
which means under misreporting, $\pi^{\hat{\theta}}(b)$ instead of $\pi^\theta$ determines the probability of the commitment incurring $b$, and then the follower will best respond $BR^\theta(b)$ according his true type. With the functions defined above, we are now ready the give the definition of the equilibrium in the proposed \texttt{sigBSG}. 
\begin{definition}
	For the \texttt{sigBSG}, a pair $(\bpi, \psi_\theta)$ is a signaling Bayesian Stackelberg equilibirum (sigBSE) if 
	\begin{align*}
		&\bpi\in \argmax\sum_{\theta\in \Theta}\mu(\theta)V(\bpi,(\psi_{\theta}(\bpi);\theta))
	\end{align*}
	where $\psi_\theta$ is defined in \eqref{eq:new_report}.  
\end{definition}

Seemingly, by adopting this probabilistic language, we have made the problem even more complicated, as now we have to search for the optimal signaling (probability measure) in the infinite-dimensional space $\Pi\subset(\Delta(\Xi))^K$. In fact, we demonstrate in the following that the leader's optimal signaling can always be found in a finite-dimensional subspace of $\Pi$. For each type $\theta\in \Theta$ and action $j\in [N]$, the set $\Xi^\theta_j:=\{b\in \Xi|j\in BR^\theta(b)\}$ includes all beliefs for which $j$ is the best response for the type-$\theta$ follower. Let $\boldsymbol{j}:=(j^k)_{k\in {K}}\in [N]^K$ be a tuple that specifies one action $j^k$ for each type $\theta_k\in \Theta$. Then for each tuple $\boldsymbol{j}$, let $\Xi_{\boldsymbol{j}}:=\cap_{k\in K}\Xi^{\theta_k}_{j^k}$ be the polytope such that $j^k$ is the best response for the corresponding type $\theta_k$. We note that $\Xi_{\boldsymbol{j}}$ is a convex polytope because each $\Xi^{\theta_k}_{j^k}$ is a polytope: best response can be described by a set of linear constraints. Any probability distribution over posterior beliefs in this intersection corresponds to a signaling scheme such that for every type $\theta_k$, the follower has no incentive to deviate from the corresponding action $j^k$. Let $\hat{\Xi}$ be the set of beliefs defined as $\hat{\Xi}:=\cup_{\boldsymbol{j}\in [N]^K}\operatorname{Vert}(\Xi_{\boldsymbol{j}})$, where $\operatorname{Vert}(\cdot)$ denotes the set of vertices (extreme points) of the given polytope. Finally, we define the following set of consistent distributions 
\begin{align*}
	&\widehat{\Pi}:=\{\boldsymbol{\pi}\in(\Delta(\hat{\Xi}))^K | \sum_{b\in\hat{\Xi}}\pi^\theta(b)b(i)=\sum_{b\in\hat{\Xi}}\pi^{\theta'}(b)b(i), \theta\neq \theta', i\in[M] \}.
\end{align*}
We note that $\widehat{\Pi}$ is a convex polytope since it can be seen as the intersection of $(\Delta(\hat{\Xi}))^K$ and a finite number of half-spaces, where the associated hyperplanes are specified by the linear constraints $\sum_{b\in\hat{\Xi}}\pi^\theta(i)b(i)=\sum_{b\in\hat{\Xi}}\pi^{\theta'}(i)b(i)$. We include more details in \Cref{apped:belief}.

One of our key findings is that even though the leader's valid signaling commitments live in an infinite-dimensional space, she can restrict her attention to $\widehat{\Pi}$ without deteriorating her expected utility. Formally, we state this result in the following lemma.
\begin{lemma}\label{lem:reduce1}
	Given the follower's optimal reporting strategy defined in \eqref{eq:new_report}, for any $\bpi\in \Pi$, there exists a $\hat{\bpi}\in \widehat{\Pi}$ such that 
	\begin{align}\label{eq:hat_pi_reduce}
		V(\bpi,(\psi_{\theta}(\bpi);\theta))=V(\hat{\bpi},(\psi_{\theta}(\hat{\bpi});\theta)),\quad \text{for all } \theta\in \Theta.
	\end{align}
\end{lemma}  
\begin{proof}[Proof of Sketch]
	The proof of this lemma is based on the following observation: any posterior belief $b$ must belong to $\Xi_{\boldsymbol{j}}$ for some $\boldsymbol{j}\in[N]^K$, and hence, it can be represented as the convex combination of elements of $\hat{\Xi}$. We denote such convex combination by $\omega^b\in \Delta(\hat{\Xi})$, where for $b'\in \hat{\Xi}$, $\omega^b(b')$ is coefficient corresponding to $b'$. Given $\bpi$, we define a new signaling $\hat{\bpi}$ as $\hat{\pi}^{\theta}(b):=\sum_{\tilde{b}\in \supp(\pi^\theta)}\pi^\theta(\tilde{b})\omega^{\tilde{b}}(b)$. Since $\bpi$ is consistent, it can bee shown that $\hat{\bpi}$ is also consistent, hence $\hat{\bpi}\in\widehat{\Pi}$. 

Moreover, for the $\hat{\pi}$ defined above, we have 
\begin{align}
	&V(\boldsymbol{\pi},(\hat{\theta};\theta))=V(\hat{\bpi},(\hat{\theta};\theta)).\label{eq:V_eq}\\
	&U(\bpi,(\hat{\theta};\theta))=U(\hat{\bpi},(\hat{\theta};\theta)),\label{eq:U_eq}
\end{align}
which implies that $\hat{\bpi}$ induces the same optimal reporting strategy as $\bpi$ does (see \eqref{eq:U_eq}). Hence, \eqref{eq:hat_pi_reduce} can be obtained directly from \eqref{eq:V_eq}. 
\end{proof}  
\Cref{lem:reduce1} is of great importance in that we now can study the leader's expected utility with respect to the prior distribution over types in the finite-dimensional convex polytope $\widehat{\Pi}\in[0,1]^{K|\hat{\Xi}|}$, which is much simpler to analyze than in $\Pi$. Formally, we let $V^\mu(\bpi)$ be the expected utility defined as $V^\mu(\bpi):=\sum_{\theta\in \Theta}\mu(\theta)V(\bpi,(\psi_{\theta}(\bpi);\theta))$. In the following, we prove that $V^\mu(\bpi)$ is piece-wise linear on $\widehat{\Pi}$, which paves the way for the characterization of the equilibrium. 

First, we introduce a partition of $\widehat{\Pi}$, where each partitioned region is also a convex polytope. For $l,k\in [K]$, we defined
\begin{align}
    \mathcal{P}^{l}_k:=\{\bpi\in \widehat{\Pi}| \psi_{\theta_k}(\bpi)=\theta_l\},
\end{align}
which is a subset of $\widehat{\Pi}$, including all signaling commitments that incentivizes the type $\theta_k$ follower to report $\theta_l$. According to the definition of the optimal reporting in \eqref{eq:new_report} and related quantities, it can be seen that $\mathcal{P}^{l}_k$ is a convex polytope.  Moreover, for any surjective mapping $\gamma:\Theta\rightarrow\Theta$, let 
\begin{align}
    \mathcal{P}^\gamma:=\cap_{k\in [K]}\mathcal{P}^{\gamma(k)}_k,
\end{align}
which is naturally also a convex polytope, including all signalings for which $\gamma(\theta)$ is the optimal reporting for the type $\theta$ follower. Let $\Gamma$ be the set of all such surjective mappings, for which $\mathcal{P}^\gamma$ is non-empty, and the collection of these subsets corresponds to a partition of $\widehat{\Pi}$. As an important remark, we point out that due to the tie-breaking rule in the definition of  optimal reporting \eqref{eq:new_report}, $\mathcal{P}^{l}_k$  is not necessarily a closed convex polytope, and therefore, $\mathcal{P}^\gamma$ is not closed either, which may lead to a discontinuous $V^\mu$, as we shall see more clearly in the proof of the following lemma.

\begin{lemma}\label{lem:linear}
 $V^\mu$ is a piece-wise linear  function on $\widehat{\Pi}$, which is linear on every $\mathcal{P}^\gamma, \gamma\in \Gamma$.
\end{lemma}
\begin{proof}[Proof of Sketch]
	For any arbitrary $\gamma\in \Gamma$, for $\bpi\in\mathcal{P}^\gamma$, the optimal reporting is specified by the mapping $\gamma$. Hence, when restricting $V^\mu$ onto $\mathcal{P}^\gamma$, we have 
	\begin{align*}
		V^\mu|_{\mathcal{P}^\gamma}(\bpi):&=\sum_{\theta\in \Theta}\mu(\theta)V(\bpi,(\gamma(\theta);\theta))\\
		&=\sum_{\theta\in \Theta}\mu(\theta)\sum_{b\in \hat{\Xi}}\bpi^{\gamma(\theta)}(b)V^\theta(b),
	\end{align*} 
	which is essentially a sum of inner products between two vectors $(\bpi^{\gamma(\theta)})_{b\in\hat{\Xi}}$ and $(V^\theta)_{b\in\hat{\Xi}}$.
\end{proof}

Recall that $\mathcal{P}^\gamma$ may not be a closed convex polytope, hence we cannot conclude that $V^\mu$ is a continuous function. To illustrate the discontinuity,  we consider an example where there exists $\mathcal{B}(\bpi_1)$, an $\epsilon$-neighborhood of $\bpi_1$ in $\widehat{\Pi}$ (with respect to some norm, e.g. $\ell_2$-norm), and there are some elements belonging to another partition. Suppose $\bpi_1\in \mathcal{P}^{\gamma_1}$ and $\bpi_2\in \mathcal{B}(\bpi_1)\cap\mathcal{P}^{\gamma_2}$, since $\gamma_1\neq \gamma_2$, there exists at least one $\theta\in \Theta$ such that $\gamma_1(\theta)\neq \gamma_2(\theta)$, for which we let $\alpha=\gamma_1(\theta)$ and $\beta=\gamma_2(\theta)$. In this case, 
\begin{align*}
	V^\mu(\bpi_1)-V^\mu(\bpi_2)=\sum_{\theta\in \Theta}\mu(\theta)\sum_{b\in \hat{\Xi}}[\bpi_1^{\gamma_1(\theta)}(b)-\bpi_2^{\gamma_2(\theta)}(b)]V^\theta(b).
\end{align*}  
No matter how small $\epsilon$ is, the difference $[\bpi_1^{\alpha}(b)-\bpi_2^{\beta}(b)]V^\theta(b)$ cannot be bounded by an $\epsilon$ term. Nevertheless, for the bounded piecewise linear function $V^\mu$ defined on a compact convex polytope $\widehat{\Pi}$, its supremum is a finite value. Even though the supremum may not be achieved, that is, the exact equilibrium may not exist; it is always possible to find an approximate: a $\bpi^\epsilon\in \widehat{\Pi}$ satisfying 
\begin{align}\label{eq:eps_V}
	\sup_{\bpi\in \widehat{\Pi}}V^\mu(\bpi)-V^\mu(\bpi^\epsilon)<\epsilon,
\end{align}
which is the $\epsilon$-equilibrium of the sigBSG. We state our result more formally in the following.
\begin{theorem}
	For the \texttt{sigBSG}, the supremum of the expected utility is a finite value and can be approximated by an $\epsilon$-sigBSE, defined in \eqref{eq:eps_V}. 
\end{theorem}
\begin{proof}
	For any sequence $\{\bpi_k\}_{k\in \N}\in \widehat{\Pi}$ that converges to the supremum value, that is
	$\lim_{k\rightarrow\infty}V^{\mu}(\bpi_k)=\sup_{\bpi\in \widehat{\Pi}}V^\mu(\bpi)$, there always exists a $\gamma\in \Gamma$, such that $\mathcal{P}^\gamma$ include a subsequence of $\{\bpi_k\}_{k\in \N}$. This is because if every $\mathcal{P}^\gamma$ only comprises a finite part of the sequence, then the original sequence would be finite, as there are only finite partitions, contradicting the definition of supremum. Without loss of generality, we still denote the subsequence by $\{\bpi_k\}_{k\in \N}$, and a $\bpi_K$ with sufficiently large $K\in \N$ is naturally an $\epsilon-$ equilibrium. 
\end{proof}
It should be noted that in basic Bayesian Stackelberg game, the notion of approximate equilibrium is generally not used in the related literature\footnote{As a solution concept per se, the approximate equilibrium may not carry significant meaning in Stackelberg games \citep{fudenberg}. However, we must point out that from the perspective of algorithmic game theory, the approximability results are important both in theory, and practice \citep{Algorithmic2007}.}, as the communication between players and the information asymmetry is one-sided: the leader makes her commitment with her incomplete knowledge about the follower, and the follower is purely receiver. Hence, from the leader's perspective, when announcing the commitment, she must take into account the possible best responses of followers of various types, which is generally described by a set of linear constraints (see the programming-based computation approaches in \citep{conitzer06compute_SG,conitzer16mixed_sg}). However, in our proposed Stackelberg game model, the approximate equilibrium arises from the two-way communication and double-sided information asymmetry, where the leader (the follower) reveals partial information regarding her realized action (his private type), which is hidden from the other side. In the presence of the follower's misreporting, as argued in \Cref{lem:linear} and its subsequent discussion, the leader may not be able to achieve the supremum in the end, and the approximate equilibrium is the best she can hope.      

In the following, we pinpoint the cause of the subtle difference in solution concepts between the baseline Bayesian Stackelberg model and  our proposed model. Considering the geometric properties of $\widehat{\Pi}$ and the piece-wise linearity of $V^\mu$, we adopt a mathematical programming viewpoint. Moreover, the geometric intuition, in fact, leads to a finite-time algorithm for searching the $\epsilon-$Stackelberg equilibrium. Similar to our partition $\mathcal{P}^\gamma$, in the baseline Bayesian Stackelberg model, each follower's pure strategy also corresponds to a convex polytope region in the space mixed strategy, and when restricting the leader's utility to an arbitrary polytope, the maximum value is achieved at some extreme point of the polytope. This is not surprising as seeking the optimal leader's strategy is equivalent to solving LPs, and naturally, the solution is some vertex of the feasible domain, and we refer the reader to \citep{conitzer06compute_SG} for more details and to \citep{letchford09br_learn} for a visualization.  

However, due to the tie-breaking rule in the follower's optimal reporting strategy, $\mathcal{P}^\gamma$ may be open for some $\gamma\in \Gamma$. Even though $V^\mu$ is a linear function on $\mathcal{P}^\gamma$, the supremum of $V^\mu|_{\mathcal{P}^\gamma}$ may not be achieved within the region. It is only guaranteed that the supremum value of $V^\mu|_{\mathcal{P}^\gamma}$ can be attained by some extreme point of the closure of $\mathcal{P}^\gamma$, denoted by $\overline{\mathcal{P}^\gamma}$. Therefore, by applying the same argument to other convex polytopes $\mathcal{P}^{\gamma}, \gamma\in \Gamma$, it can be seen that the supremum value of $V^{\mu}$ over $\widehat{\Pi}$ may not be achievable. On the other hand, despite the possible non-existence of exact equilibrium, it is always possible to approximate it by finding a point sufficiently close to the desired extreme point of the closure. Formally, we have the following result. For any $\epsilon>0$, let $\delta=\frac{\epsilon}{\|V^\theta\|_2}$, where $\|V^\theta\|_2$ denotes the $\ell_2$-norm of the vector $(V^\theta)_{b\in\hat{\Xi}}$. We define $\Pi^\delta$ to be a finite set of $\bpi$ as follow: for every $\gamma\in \Gamma$ and every $\bpi$ that is an extreme point of $\mathcal{P}^\gamma$, if $\bpi\in \mathcal{P}^\gamma$, then $\bpi\in \Pi^\delta$. If not, let $\Pi^\delta$ include an arbitrary element $\bpi'$ in $\mathcal{P}^\gamma$ such that $\|\bpi-\bpi'\|_2\leq \delta$. In other words, the element in the set $\Pi^\delta$ is either an extreme point of $\mathcal{P}^\gamma$ for some $\gamma$ or a point that is close enough to some extreme point of the closure $\overline{\mathcal{P}^\gamma}$. Notice that $|\Pi^\delta|<\infty$, since there are only finitely many partitions, the leader at most needs to compute the utility at $|\Pi^\delta|$ points to identify the $\epsilon-$equilibrium, which is formally stated below.  
 
\begin{theorem}
	For the \texttt{sigBSG}, and any $\epsilon>0$, let $\Pi^\delta$ be defined as above. There exists an $\epsilon-$equilibrium in $\Pi^\delta$, which can be found in $O\left({K^{K+2}\choose K|\hat{\Xi}|}K^K\right)$. 
\end{theorem} 
  The following presents a direct corollary from the geometric characterization of the function $V^\mu$, which shows that the leader does benefit from the signaling (the equality holds when the revelation principle \citep{myerson82pa} holds, which is the case here). This is because any BSE can be viewed as a special signaling commitment with uninformative signaling, as shown in the running example. 
\begin{corollary}
	Let the supremum value of $V^\mu$ be denoted by $V_{\text{sigBSE}}$ and the leader's equilibrium payoff under vanilla BSE by $V_{\text{BSE}}$, then $V_{\text{sigBSE}}\geq V_{\text{BSE}}.$
\end{corollary}

\section{Online Ignorant Learning for Optimal Signaling}
What we have discussed so far falls within the realm of information design or Bayesian persuasion with privately informed receiver \citep{emir11persuasion,kolo17private}, where the leader already acquires enough information, including possible follower's types as well as the prior distribution $\mu^*$, to compute the sigBSE offline. However, in reality, the leader may not even know the prior distribution (an ignorant leader), and hence she has to learn a signaling commitment with a satisfying online performance. In this section, we aim to provide an online learning framework that enables the leader to adaptively adjust her signaling commitment based on her past experience, and we show that the proposed learning method guarantees that the online performance of the learned commitment is comparable to the equilibrium payoff under sigBSE, the optimal expected utility.

\subsection{Repeated Play with Signaling Commitment}
Consider a repeated Stackelberg game over a time horizon of length $T$, where the leader has the perfect recall of historical play while the follower is memoryless. The assumption implies that only the leader will adapt herself to the changing environment, whereas the follower always behaves myopically, which we believe reflects the market competition: leading firms are usually those who have survived numerous competitions, and enjoy longer histories, whereas emerging start-ups come and go.

Mathematically, after each round of play, the current follower is effectively removed from the game, and Nature, according to the distribution $\mu^*$, randomly draws a new follower, who will play the game with the long-standing leader in the next round. The $t$-th round includes five stages: 1) Nature creates a new follower of type $\theta_t$; 2) the leader announces her signaling commitment $\sigma_t=(\mathbf{x}_t,\mathbf{C}_t)$; 3) the leader samples an action $i_t$, which cannot be observed by the follower;  4)the follower reports his type according to the optimal reporting strategy $\psi_{\theta_t}(\sigma_t)$, defined in \eqref{eq:report}; 5) the leader sends a signal according to her commitment $\sigma_t$; 6) after observing the signal, the follower plays his best response with respect to the induced posterior belief; 6) the reward is revealed to the leader.

\paragraph{Information Structures} From leader's perspective, the information she acquires regarding the historical play is of vital importance when designing online learning algorithms \citep{tao_info}. Therefore, before detailing the proposed learning method, we first clarify what information is available to the leader. 

Directly from the repeated Stackelberg game model described above, we can see that the leader is aware of the following: her signaling commitments $\sigma_{t\in [T]}$, her realized actions $i_{t\in [T]}$, follower's reported type $\hat{\theta}_{t\in [T]}$ and follower's best response $j_{t\in [T]}$, based on which we can further assume that the leader can observe the follower's true type $\theta_{t\in [T]}$. In fact, it is always viable to unearth the follower's true type. If the leader adopts incentive-compatible commitment specified by \eqref{eq:ic} in the following linear programming, then the reported type equals the true type.
\begin{align}
    \max_{\sigma=(\mathbf{x},\mathbf{C})}\quad & \sum_{\theta\in \Theta} \mu^*(\theta) \sum_{i\in [M], j\in [N]} C^{\theta}_{ij}L_{ij}\tag{sigLP}\label{eq:sigbse}\\
    \st \quad & \sum_{i\in [M]}\mathbf{x}_i=1,\nonumber\\
    & \sum_{j\in [N]}C^\theta_{ij}=\mathbf{x}_i, \text{ for all } i, \theta\nonumber\\
    & C^\theta_{ij}\geq 0, \text{ for all } i,j,\theta\nonumber\\
    & \sum_{i\in [M]} {C}^\theta_{ij}F^\theta_{ij}\geq \sum_{i\in [M]}{C}^\theta_{ij}F^\theta_{ij'}, \text{ for all } j\neq j',\label{eq:obedience}\\
    & \sum_{i\in [M]}\sum_{j\in [N]} C^{\theta}_{ij}F^\theta_{ij}\geq \nonumber\\
    &\qquad\sum_{j\in [N]}\max_{j'} \sum_{i\in [M]}C^{\hat{\theta}}_{ij}F^\theta_{i,j'}, \text{ for all } \hat{\theta}\neq \theta. \label{eq:ic}
\end{align}

 On the other hand, if only an obedience signaling commitment is utilized, though the follower may misreport, the leader can immediately realize this follower's deception, as her recommendation is ignored by the follower. Hence, the leader can do a reverse engineering to \eqref{eq:report}, based on her commitment $\sigma_t$ and the follower's response $j_t$. Even though chances are that the leader may find multiple candidate solutions to the inverse problem or it so happens that the deceitful follower is still  obedient, we do not consider these cases in our learning model. To sum up, in this repeated play, the history is $H_t:=\{\sigma_{\tau\in [t]}, \hat{\theta}_{\tau\in [t]}, \theta_{\tau\in [t]}, i_{\tau\in[t]}, j_{\tau\in[t]}\}$.

 As a high-level summary of our repeated Stackelberg game model, we point out that at each round of the repeated play, the leader moves first, where the commitment is based on the history, then the follower of an unknown type best responds in the sense that he both reports optimally (see \eqref{eq:report}) and plays optimally.  As we have mentioned \Cref{sec:intro}, the question we ask is whether the leader still enjoys the first-mover advantage, or, equivalently, does there exist an online learning algorithm that guarantees the online performance? 

To evaluate the performance of the learning process, we need a proper metric. Given an online learning algorithm $\mathfrak{L}$, its online performace is defined as $V_{\mathfrak{L}}(T):=\sum_{t\in[T]}L(i_t,j_t)$. We note that $V_{\mathfrak{L}}(T)$ can be viewed as a random function of the learning algorithm $\mathfrak{L}$ as well as $T$.  Let the leader's sigBSE equilibrium payoff under the prior distribution $\mu^*$ be $\operatorname{OPT}(\mu^*)$, then the optimal total expected payoff is $\operatorname{OPT}(\mu^*)\cdot T$. Therefore, we can evaluate the learning algorithm using  the following performance gap 
\begin{align}\label{eq:gap}
    \operatorname{Gap}_{\mathfrak{L}}(\mu^*,T):=\operatorname{OPT}(\mu^*)\cdot T- V_{\mathfrak{L}}(T).
\end{align}
Then the question reduces to whether there exists an algorithm $\mathfrak{L}$ such that $\operatorname{Gap}_{\mathfrak{L}}(\mu^*,T)$ grows sublinearly with respect to $T$ with a high probability, implying that the average payoff obtained by the leader is as good as the sigBSE payoff.

We answer this question affirmatively, and the rest of this section is devoted to developing such online learning framework.  We start with a simple case, where the leader is allowed to use incentive-compatible commitment [see \eqref{eq:ic} in \eqref{eq:sigbse}], and we show that the follow-the-leader policy \citep{csaba_bandit} suffices to provide a desired online performance. 

\subsection{Follow-the-Incentive-Compatible-Leader} Due to the constraints in \eqref{eq:ic}, if  the leader commits to an incentive compatible $\sigma_t=(\mathbf{x}_t,\mathbf{C}_t)$ in every round (referred to as the IC-leader), the follower will always truthfully reports his type and follow the leader's recommendation. In this case, the leader's expected utility under a signaling commitment $\sigma$ when the follower is of type $\theta$ can be easily computed as
\begin{align}\label{eq:value_sig}
    V(\sigma,\theta):=\sum_{i\in [M], j\in [N]}C^\theta_{ij}L_{ij}.
\end{align}
One simple principle, referred to as follower-the-leader in the literature \citep{csaba_bandit}, to prescribe an online learning process is the idea of optimality in hindsight: the signaling commitment at the next round $\sigma_t$ should be the optimal one in hindsight up to the current stage, which maximizes the cumulative utility in the past had it been applied. In other words, according to the current history $H_t$, $\sigma_{t+1}$ is determined by
\begin{align}
\label{eq:ftl}
    \sigma_{t+1}:=\argmax_{\sigma} \sum_{\tau\in [t]}V(\theta_t,\sigma),\tag{FTL-IC}
\end{align}
where $\sigma$ is subject to the constraints in \eqref{eq:sigbse}. Since the commitment in each round satisfies the IC constraints, we refer to the learning algorithm in \eqref{eq:ftl} as follow-the-incentive-compatible-leader. 

Furthermore, we denote $\mu_t$ the empirical distribution of $\{\theta_{\tau\in [t]}\}$ at time $t$, that is $\mu_t(\theta):=\sum_{\tau\in [t]}\mathbb{I}(\theta_t=\theta)/t$, where $\mathbb{I}(\cdot)$ is the indicator function. Then, it can be seen that $\sigma_{t+1}$ is the solution to \eqref{eq:sigbse} with $\mu^*$ being replaced by $\mu_t$. In other words, the leader always chooses the signaling commitment specified by the sigBSE, with the prior being $\mu_t$. According to the concentration bound, $\mu_t$ converges to $\mu^*$ in point-wise with high probability (i.i.d. sampling) \citep{compress}, and naturally, \eqref{eq:ftl} gives a desired online performance. We state our result formally in the following theorem.
\begin{theorem}\label{thm:ftl}
    For $t\in [T]$ and any $\mu^*\in \Delta(\Theta)$, with probability at least $1-T^{1-\frac{3\sqrt{M}}{56}}-T^{-8M}$, the learning algorithm \eqref{eq:ftl} satisfies 
    \begin{align*}
        \operatorname{Gap}(\mu^*, T)\leq O(\sqrt{MT}(1+2\sqrt{\log T})).
    \end{align*}
\end{theorem}

\subsection{Online Ignorant Learning}
Note that the success of \eqref{eq:ftl} relies on the obedience \eqref{eq:obedience} and IC constraints \eqref{eq:ic}, making it easier to compute $V$. We now consider the ignorant leader, who does not impose the IC constraints when learning her commitment. As shown in the following, even though the situation becomes more complicated, results similar to \cref{thm:ftl} also hold, meaning that the first-mover advantage still holds. The gist behind this ignorant learning lies in \Cref{lem:reduce1}: without loss of performance, the leader can restrict her action space to a finite-dimensional one (and eventually a finite set). This observation is formally presented \Cref{prop:action-reduce}, and its proof follows from \Cref{lem:reduce1}
\begin{proposition}
\label{prop:action-reduce}
    For $t\in [T]$ and any sequence of the follower's type $\theta_{t\in [T]}$, let $\Pi^*:=\cup_{\gamma\in \Gamma}\operatorname{Vert}(\widehat{\Pi}\cap \mathcal{P}^\gamma)$, 
\begin{align}
    \max_{\bpi\in \Pi}\sum_{t\in [T]}V(\bpi,\psi_{\theta_t}(\bpi))=\max_{\bpi\in \Pi^*}\sum_{t\in [T]}V(\bpi,\psi_{\theta_t}(\bpi))
\end{align}
\end{proposition}
Note that $\widehat{\Pi}\cap \mathcal{P}^\gamma$ is a convex polytope in a finite-dimensional space, implying that $\Pi^*$ is a finite set. Hence, the leader action space reduces to a finite set, without incurring any performance loss. 

One key step in designing an online learning algorithm that achieves a diminishing time-average performance gap is to show that the regret of using signaling commitments prescribed by the online learning algorithm grows sublinearly, which is referred to as no-regret learning \citep{perchet14blackwell}. Formally, regret is defined as the difference between the cumulative expected utility of best-in-hindsight signaling commitments and that of the commitments produced by the learning algorithm.
\begin{definition}[Regret]
Given an online learning algorithm $\mathfrak{L}$, the resulting regret in the repeated \texttt{sigBSG} is given by 
    \begin{align*}
    \operatorname{Reg}_{\mathfrak{L}}(T):=\max_{\bpi\in \Pi}\left\{\sum_{t\in [T]}V(\bpi,\psi_{\theta_t}(\bpi))-\E[V(\bpi_t,\psi_{\theta_t}(\bpi_t))]\right\},
\end{align*}
where the expectation is on the randomness of the online algorithm $\mathfrak{L}$. 
\end{definition}

 From the leader's perspective, the problem of learning to signal reduces to online decision-making with a finite action space in an adversarial environment, where the adversarial nature stems from the leader's unawareness of the follower's true type. Therefore, any no-regret learning algorithm, such as follow-the-leader \citep{csaba_bandit} and much broadly speaking, those based on Blackwell approachability \citep{Tao_blackwell,abernethy11approach}, can all be considered as a black-box in designing the learning algorithm. The following theorem is based on the seminal result regarding no-regret learning in the adversarial bandit problem \citep{perchet14blackwell}.  
\begin{theorem}
\label{thm:regret}
    For $t\in [T]$ and every sequence of the follower's type $\theta_{t\in [T]}$, there exists an online learning algorithm $\mathfrak{L}$ such that
    \begin{align}\label{eq:no_regret}
        \operatorname{Reg}_{\mathfrak{L}}(T)\leq O\left(\sqrt{TM\log(K^2N^2+KM)}\right).
    \end{align}
\end{theorem}

Finally, with the concentration inequality regarding i.i.d. samples $\mu_t$ , we can show in the following theorem that even though the ignorant leader does not impose IC constraints, she can still learn a signaling commitment that guarantees her average cumulative payoff is as good as her equilibrium payoff in the limit. 
\begin{theorem}\label{thm:adver}
    For $t\in [T]$ and any $\mu^*\in \Delta(\Theta)$, there exists an online learning algorithm satisfying \eqref{eq:no_regret}, for which with probability at least $1-T^{1-\frac{3\sqrt{M}}{56}}-T^{-8M}$, 
    \begin{align*}
        \operatorname{Gap}(\mu^*, T)\leq O(\sqrt{TM\log(K^2N^2+KM)}(1+2\sqrt{\log T})).
    \end{align*}
\end{theorem}

\section{Conclusion}\label{sec:conclu}
This paper investigates a novel Stackelberg game model, where the leader gains the extra commitment ability: designing a signaling device. In order to coordinate the follower's moves by signaling, the leader needs to signal differently to followers with different types, which requires the follower to report his type to the leader. In the proposed game model \texttt{sigBSG}, both the leader and the follower have asymmetric information, which enables  both sides to strategically reveal information to the other. With the first-mover advantage, the leader can offset the follower's information advantage by optimally designing the signaling device, which brings her higher expected utility compared with the Bayesian Stackelberg equilibrium payoff. Besides the theoretical results, our probabilistic interpretation of the leader's commitment, which is built upon the geometry of the belief space,  greatly simplifies the analysis and provides a geometric intuition of the equilibrium characterization.
\newpage
\begin{appendix}
	\section{Geometric Characterizations of Belief Spaces }\label{apped:belief}
	In our arguments in the main text, we have mentioned that belief spaces $\Xi_j^\theta$ and $\Xi_{\mathbf{j}}$ are convex polytopes, which further leads to the fact that $\hat{\Pi}$ is a convex polytope, and so are its partitions $\mathcal{P}^l_k$. In this appendix, we prove our claims in mathematical terms.
	
	We first note that since the leader's pure strategy set $[M]$ is finite, the induced belief $b$ belongs to the probability simplex $\Delta_M\in \R^M$, a finite-dimensional vector space. Therefore, for every element $b$ in the set $\Xi_j^\theta:=\{b\in \Xi|j\in BR^\theta(b)\}$, it satisfies the following linear inequalities$$\sum_{i\in [M]}b(i)F^\theta(i,j)\geq \sum_{i\in [M]}b(i)F^\theta(i,j'), \quad\text{for all } j'\in [N],$$  which implies that $\Xi_j^\theta$ is a convex polytope. By the definition, $\Xi_{\mathbf{j}}:=\cap_{k\in K}\Xi^{\theta_k}_{j^k}$, $\Xi_{\mathbf{j}}$ is an intersection of $K$ convex polytopes and hence is also a convex polytope.
	
	Since $\Xi_{\mathbf{j}}$  is a convex polytope in a finite-dimensional space, its vertices (extreme points) must be finite; that is, $\hat{\Xi}$ is a finite set. This implies that $\Delta(\hat{\Xi})$ is a probability simplex in a finite-dimensional space, and hence $\hat{\Pi}$ is also of finite dimension, in which every element $\pi$ is a concatenation of $(\pi^\theta)_{\theta\in \Theta}$. To prove that $\hat{\Pi}$ is a convex polytope, we let $W$ be a $M\times|\hat{\Xi}|$ matrix with each column being one element in $\hat{\Xi}$, then the constraint $\sum_{b\in\hat{\Xi}}\pi^\theta(b)b(i)=\sum_{b\in\hat{\Xi}}\pi^{\theta'}(b)b(i)$ can be rewritten as $W\pi^\theta-W\pi^{\theta'}=0$.  Finally, suitably defining $\Phi$ a block matrix of dimension $KM\times|\hat{\Xi}|$, then the constraints in the definition of $\widehat{\Xi}$ is $\Phi\boldsymbol{\pi}=0$, which can be safely rewritten as $\Phi\boldsymbol{\pi}\geq 0$, as $\boldsymbol{\pi}$ lies within the product of simplex spaces. Therefore, $\hat{\Xi}$ is a convex polytope in a finite-dimensional space.
	
	What remains is to show that $\mathcal{P}_k^l:=\{\bpi\in \widehat{\Pi}| \psi_{\theta_k}(\bpi)=\theta_l\}$ is a convex polytope. According to the tie-breaking rule, the reported type should be the $\theta_l$ with the smallest index that satisfies \eqref{eq:new_report}. In other words, for any $\bpi\in \mathcal{P}_k^l$, it satisfies the following linear inequalities
	\begin{enumerate}[1)]
		\item if $1\leq k<l\leq K$:
	\begin{align*}
		&\sum_{b\in \hat{\Xi}}\pi^{\theta_l}(b)F^{\theta_k}(b)>\sum_{b\in \hat{\Xi}}\pi^{\theta_q}(b)F^{\theta_k}(b), \quad\text{for all } 1\leq q\leq l-1, \\
		&\sum_{b\in \hat{\Xi}}\pi^{\theta_l}(b)F^{\theta_k}(b)\geq\sum_{b\in \hat{\Xi}}\pi^{\theta_q}(b)F^{\theta_k}(b), \quad\text{for all } l+1\leq q\leq K ;
	\end{align*}
	\item if $1\leq l<k\leq K$:
	\begin{align*}
		&\sum_{b\in \hat{\Xi}}\pi^{\theta_l}(b)F^{\theta_k}(b)>\sum_{b\in \hat{\Xi}}\pi^{\theta_q}(b)F^{\theta_k}(b), \quad\text{for all } 1\leq q\leq l-1, \\
		&\sum_{b\in \hat{\Xi}}\pi^{\theta_l}(b)F^{\theta_k}(b)>\sum_{b\in \hat{\Xi}}\pi^{\theta_k}(b)F^{\theta_k}(b),\\
		&\sum_{b\in \hat{\Xi}}\pi^{\theta_l}(b)F^{\theta_k}(b)\geq\sum_{b\in \hat{\Xi}}\pi^{\theta_q}(b)F^{\theta_k}(b), \quad\text{for all } q\neq k, l+1\leq q\leq K ;
	\end{align*}
	\item if $1\leq l=k\leq K$:
	\begin{align*}
		&\sum_{b\in \hat{\Xi}}\pi^{\theta_l}(b)F^{\theta_k}(b)=\sum_{b\in \hat{\Xi}}\pi^{\theta_k}(b)F^{\theta_k}(b),\\
		&\sum_{b\in \hat{\Xi}}\pi^{\theta_l}(b)F^{\theta_k}(b)\geq\sum_{b\in \hat{\Xi}}\pi^{\theta_q}(b)F^{\theta_k}(b), \quad\text{for all } 1\leq q\leq K .
	\end{align*}
	\end{enumerate}
 	Therefore, we conclude that $\mathcal{P}_k^l$ is also a convex polytope but not necessarily a closed one, and further, $\mathcal{P}^\gamma:=\cap_{k\in [K]}\mathcal{P}^{\gamma(k)}_k$ , as an intersection of finite convex polytopes, is also a convex polytope,

	\section{Full Proofs}
	In this appendix, we provide the reader with detailed proofs to our results in the main text.
	\begin{proof}[Proof of \Cref{lem:reduce1}]
		We begin the proof with the following fact: since $\hat{\Xi}$ is a convex polytope, any posterior belief $b$ must belong to $\Xi_{\boldsymbol{j}}$ for some $\boldsymbol{j}\in[N]^K$, and hence, it can be represented as the convex combination of elements of $\hat{\Xi}$. We denote such convex combination by $\omega^b\in \Delta(\hat{\Xi})$, where for $b'\in \hat{\Xi}$, $\omega^b(b')$ is coefficient corresponding to $b'$. Given $\bpi$, we define a new signaling $\hat{\bpi}$ as 
\begin{align*}
	\hat{\pi}^{\theta}(b):=\sum_{\tilde{b}\in \supp(\pi^\theta)}\pi^\theta(\tilde{b})\omega^{\tilde{b}}(b).
\end{align*}
Since $\bpi$ is consistent, $\hat{\bpi}$ is also consistent, as shown below
\begin{align*}
	\sum_{b\in \hat{\Xi}}\hat{\pi}^{\theta}(b)b(i)&=\sum_{b\in \hat{\Xi}}\sum_{\tilde{b}\in \supp(\pi^\theta)}\pi^\theta(\tilde{b})\omega^{\tilde{b}}(b)b(i)\\
	&=\sum_{\tilde{b}\in \supp(\pi^\theta)}\pi^\theta(\tilde{b})\tilde{b}(i)\\
	&=\sum_{\tilde{b}\in \supp(\pi^{\theta'})}\pi^{\theta'}(\tilde{b})\tilde{b}(i)\\
	&=\sum_{b\in \hat{\Xi}}\hat{\pi}^{\theta'}(b)b(i),
\end{align*}
which implies that $\bpi^*$ is also a valid signaling commitment. 

Given a belief $b'$, let $\boldsymbol{j}$ be the tuple specifying the best response under this belief. At each $b\in \operatorname{Vert}(\Xi_{\boldsymbol{j}})$, the follower of type $\theta$ plays $BR^\theta(b)$, and hence, the following holds
\begin{align*}
	\sum_{i\in [M]}b(i)L(i,BR^\theta(b))=\sum_{i\in [M]}b(i)L(i,BR^\theta(b')),
\end{align*}  
which further leads to 
\begin{align*}
	\sum_{b\in \operatorname{Vert}(\Xi_{\boldsymbol{j}})}\omega^{b'}(b)V^\theta(b)&=\sum_{b\in \operatorname{Vert}(\Xi_{\boldsymbol{j}})}\omega^{b'}(b)\sum_{i\in[M]}b(i)L(i,BR^\theta(b))\\
	&=\sum_{b\in \operatorname{Vert}(\Xi_{\boldsymbol{j}})}\omega^{b'}(b)\sum_{i\in [M]}b(i)L(i,BR^\theta(b'))\\
	&=\sum_{i\in [M]}b'(i)L(i,BR^\theta(b'))\\
	&=V^\theta(b').
\end{align*}
Therefore, for the leader's expected utility $V(\bpi, (\hat{\theta};\theta))$, we obtain 
\begin{align*}
	V(\bpi, (\hat{\theta};\theta))&=\sum_{b'\in \supp(\pi^{\hat{\theta}})}\pi^{\hat{\theta}}(b')V^{\theta}(b')\\
	&=\sum_{b'\in \supp(\pi^{\hat{\theta}})}\pi^{\hat{\theta}}(b')\sum_{b\in \operatorname{Vert}(\Xi_{\boldsymbol{j}})}\omega^{b'}(b)V^\theta(b)\\
	&=\sum_{b\in \hat{\Xi}}\hat{\pi}^{\hat{\theta}}(b)V^\theta(b)\\
	&=V(\hat{\bpi}, (\hat{\theta};\theta)),
\end{align*}
which proves \eqref{eq:V_eq} and a similar argument proves \eqref{eq:U_eq}.
	\end{proof}
\end{appendix}

\bibliographystyle{plainnat}
\bibliography{sigBSG_full}
\end{document}